\documentclass[10pt,conference]{IEEEtran}
\usepackage{graphicx,subfigure}
\usepackage{amsfonts,amsmath,amssymb, mathrsfs}
\usepackage{epic,eepic,eepicemu}
\usepackage{epsf}
\usepackage{epsfig}
\usepackage{graphics}
\usepackage{psfrag}

\newtheorem{theorem}{Theorem}

\newtheorem{lemma}{Lemma}

\newtheorem{remark}{Remark}

\usepackage{tikz}
\usetikzlibrary{arrows,shapes,backgrounds,snakes}

\tikzstyle{dot}=[circle,draw=gray!90,fill=gray!20,thick,inner
 sep=1pt,minimum size=16pt]
\date{}

\usepackage{array}

\begin{document}
\title{Beyond the Cut-Set Bound: Uncertainty Computations in Network Coding with Correlated Sources}
\author{Amin Aminzadeh Gohari$^*$, Shenghao Yang$^{**}$, Sidharth Jaggi$^{**}$\\\small
$^{*}$Department of Electrical Engineering, Sharif University of Technology, Iran\\ $^{**}$Institute of Network Coding, Department of Information Engineering, Chinese University of Hong Kong}

\maketitle
\begin{abstract} Cut-set bounds on achievable rates for network communication protocols are not in general tight. In this paper we introduce a new technique for proving converses for the problem of transmission of correlated sources in networks, that results in bounds that are tighter than the corresponding cut-set bounds. We also define the concept of ``uncertainty region" which might be of independent interest. We provide a full characterization of this region for the case of two correlated random variables. The bounding technique works as follows: on one hand we show that if the communication problem is solvable, the uncertainty of certain random variables in the network with respect to imaginary parties that have partial knowledge of the sources must satisfy some constraints that depend on the network architecture. On the other hand, the same uncertainties have to satisfy constraints that only depend on the joint distribution of the sources. Matching these two leads to restrictions on the statistical joint distribution of the sources in communication problems that are solvable over a given network architecture.
\end{abstract}

\section{Introduction}
Consider a directed network with a source $s$ and two sinks
$t_1$ and $t_2$.\footnote{To convey the basic ideas in the simplest way, throughout this paper we assume that there are two sources. Generalization to more than two sources (sinks) is also possible.} Suppose that the source observes i.i.d.\ copies of random variables
$X$, $Y$ jointly distributed according to $p(x,y)$. Sink $t_1$ is interested in the i.i.d. copies of $X$, while
sink $t_2$ is interested in the i.i.d. copies of $Y$. We consider the problem of
reliable transmission to fulfill the demands of both
sink nodes with probability converging to one as the number of i.i.d. observations of $X$, $Y$ grows without bound.

The cut-set bound says that if the demands of both
sinks can be fulfilled, each of the cuts that separate $s$ from $t_1$ must have capacity at least $H(X)$, each of the cuts that separate $s$ from $t_2$ must have capacity at least $H(Y)$ and each of the cuts that separate $s$ from $(t_1,t_2)$ must have capacity at least $H(X,Y)$. The cut-set bound is known to be tight when $X=(M_0,M_1)$ and $Y=(M_0,M_2)$ for some mutually independent random variables $M_0$, $M_1$, $M_2$ \cite{Ngai2004,Erez2003}.
Another case is when $X$ and $Y$ are ``linearly correlated" in the sense that one can express $X$ and $Y$ as $X = A U^m$ and $Y = B U^m$ for some random vector $U^m$, and matrices $A$
and $B$ all taking values in a given field. Without loss of generality one can assume that the rows of $A$ and $B$ are
linearly independent. By applying suitably chosen invertible
linear transformations $T_1$ and $T_2$, we can write
\begin{align*}
  T_1X & =
  \begin{bmatrix}
    A_0 \\ A_1
  \end{bmatrix}
  U^m \\
  T_2Y & =
  \begin{bmatrix}
    A_0 \\ B_1
  \end{bmatrix}
  U^m ,
\end{align*}
where the rows of $A_0$, $A_1$ and $B_1$ are linearly independent. Because the
linear transformations $T_1$ and $T_2$ are invertible, the communication task
is to transmit the common message $A_0U^m$ to both the sinks, and the private messages $A_1U^m$ and $B_1U^m$ to the two sinks. Clearly this problem reduces to the one mentioned above if $A_0U^m$, $A_1U^m$, $B_1U^m$ are mutually independent. Therefore the cut-set bound  is also tight in such cases.

However, in general when the joint distribution of $X$ and $Y$ is arbitrary the cut-set bound is not always tight. To go beyond the cut-set bound, we devise a new technique for proving converses for the problem of transmission of correlated sources over networks. We provide an example for which the cut-set bound is not tight, but the new converse is tight. Nonetheless the problem of finding joint distribution of the sources in communication problems that are solvable over a given network remains an open problem. One can refer to the several papers written on this topic for treatments of special cases of this problem (see for instance \cite{HoMedradEffrosKoetter}
-\cite{Coleman}). Some of these works discuss different settings in which separated source coding and network coding becomes either optimal or suboptimal.

At the heart of our technique lies the concept of ``uncertainty region" and how we relate it to networks. We define the uncertainty region as the set of all possible uncertainty vectors where each of these vectors are trying to capture the uncertainty of a given random variable from the perspective of different observers who have access to distinct but dependent sources. More precisely, given an arbitrary random variable $K$, a vector formed by listing the uncertainty left in $K$ when conditioned on different subsets of i.i.d.\ copies $\{X^n, Y^n\}$, i.e. $[\frac{1}{n}H(K),~\frac{1}{n}H(K|X^n),~\frac{1}{n}H(K|Y^n),~ \frac{1}{n}H(K|X^n,Y^n)]$, is called an uncertainty vector. Since the statistical dependence between the sources affects the uncertainty region in a crucial way, our discussion of correlated sources here is not an straightforward extension of the case of independent sources. Our technique also differs from those developed by Kramer et al.\ \cite{KramerSavari}, Harvey, et al.\ \cite{HarveyKleinbergLehman} and Thakor et al.\ \cite{ThakorGrantChan}, all of which concern transmission of independent sources over networks. 

The rest of the paper is organized as follows. In Section
\ref{Section:Motivation}, we motivates our new technique. Section \ref{Section:UncertaintyRegion}
contains one of the main results of this article, a complete characterization of the uncertainty region. Section
\ref{Section:Proofs} includes the proofs.
\section{Motivation}
\label{Section:Motivation}
This section motivates our technique which is based on uncertainty computations. For the ease of exposition and to convey the main ideas, discussions in this section will be \emph{quite intuitive and not rigorous}. A precise discussion will be provided later.

Let us begin with the well-known butterfly network shown in Figure \ref{fig:f1}. Assume that the source is observing $n$ i.i.d.\ repetitions of the correlated binary sources $(X,Y)$. Thus the source has a length-$n$ vector $X^n$ and the length-$n$ vector $Y^n$. The first sink is interested in recovering the $n$ i.i.d.\ repetitions of $X$ whereas the second sink is interested in recovering the $n$ i.i.d.\ repetitions of $Y$. Probabilities of error at both sinks are required to converge to zero as the number of i.i.d.\ observations of $X$, $Y$ grow without bound.
 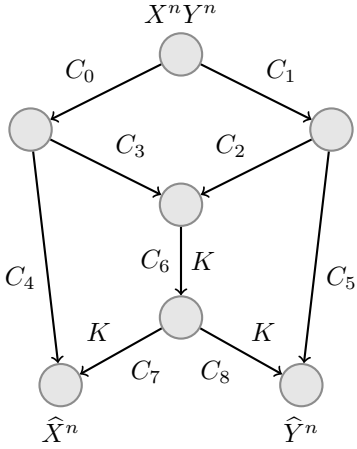
\begin{figure}
  \centering
  \begin{tikzpicture}[scale=1]
	\node[dot, label=above:$X^nY^n$] (s) at (0,0) {};
	\node[dot] (a) at (-2,-1) {}
    edge [<-,thick] node[auto] {$C_0$} (s);
	\node[dot] (b) at (2,-1) {}
    edge [<-,thick]  node[auto,swap] {$C_1$}  (s);
	\node[dot] (c) at (0,-2) {}
    edge [<-,thick] node[auto] {$C_2$} (b)
    edge [<-,thick] node[auto,swap] {$C_3$} (a);
	\node[dot] (d) at (0,-3.5) {}
    edge [<-,thick] node[auto] {$C_6$} node[auto,swap] {$K$} (c);
	\node[dot, label=below:$\widehat{X}^n$] (t) at (-1.6,-4.4) {}
    edge [<-,thick] node[auto,swap] {$C_7$} node[auto] {$K$} (d)
    edge [<-,thick] node[auto] {} node[auto] {$C_4$} (a);
	\node[dot, label=below:$\widehat{Y}^n$] (u) at (1.6,-4.4) {}
    edge [<-,thick] node[auto] {$C_8$} node[auto, swap] {$K$} (d)
    edge [<-,thick] node[auto,swap] {$C_5$} (b);
  \end{tikzpicture}
  \caption{Transmission of correlated sources over a butterfly network. The capacity of edge $i$ is $C_i$ as labeled. Assume $C_6=C_7=C_8$. $K$ is the message on edge $6$.}
  \label{fig:f1}
\end{figure}
For the sake of simplicity we restrict ourselves to networks such that the cut towards the first receiver across edges $4$ and $6$, and the cut towards the second receiver across edges $5$ and $6$, are tight; that is $C_4+C_6=H(X)$ and $C_5+C_6=H(Y)$. Let $K$ denote the random variable that is put on edge $6$ as shown in Figure \ref{fig:f1}. Using the source coding theorem and the fact that $C_4+C_6=H(X)$, one can conclude that $H(K|X^n)$ ought to be negligible if the demand of the first sink is to be fulfilled. Similarly $H(K|Y^n)$ ought to be negligible. Therefore $K$ corresponds to common randomness between $X^n$ and $Y^n$ in the sense of G\'{a}cs-K\"{o}rner \cite{GacsKorner}. This common information is equal to $\max H(T)$ where $T$ is both a function of $X$ and $Y$. For binary sources this common information is non-zero if and only if $X=Y$ or $X=1-Y$. Thus in the general case, the G\'{a}cs-K\"{o}rner common information for binary random variables is zero, implying that $\frac{1}{n}H(K)$ should be almost zero. This effectively implies that we are not using edge $6$ in communication at all. But the cuts at the two sinks were tight, implying that $C_4<H(X)$ and $C_5<H(Y)$. There is not enough rate to communicate $X^n$ and $Y^n$ through these links. This implies that the required communication demands cannot be simultaneously satisfied. Note that because even a small perturbation in the joint distribution can destroy the G\'{a}cs-K\"{o}rner common information between two random variables, a given network that supports transmission of certain correlated sources, may not support transmission of correlated sources in its immediate vicinity, a discontinuity type phenomenon.

Our second example is again based on the butterfly network of Figure \ref{fig:f2} with a passive eavesdropper on one of the nodes as shown in the figure. The eavesdropper can observe random variable $K$ but cannot tamper with any of the messages. The goal of the code is to keep the eavesdropper almost ignorant of the message of the first sink. That is, we would like to restrict our attention to those codes in which $K$ is almost independent of $X^n$. Further, assume that the cut at the second sink is tight, {\it i.e.}, $C_5+C_6=H(Y)$. We claim that one must then have $C_4\geq H(X)$, $C_6\leq H(Y|X)$, $C_5\geq I(X;Y)$. Otherwise, the sources are not transmittable.

 \begin{figure}
  \centering
  \begin{tikzpicture}[scale=1]
	\node[dot, label=above:$X^nY^n$] (s) at (0,0) {};
	\node[dot] (a) at (-2,-1) {}
    edge [<-,thick] node[auto] {$C_0$} (s);
	\node[dot] (b) at (2,-1) {}
    edge [<-,thick]  node[auto,swap] {$C_1$}  (s);
	\node[dot] (c) at (0,-2) {}
    edge [<-,thick] node[auto] {$C_2$} (b)
    edge [<-,thick] node[auto,swap] {$C_3$} (a);
	\node[dot] (d) at (0,-3.5) {$w$}
    edge [<-,thick] node[auto] {$C_6$} node[auto,swap] {$K$} (c);
	\node[dot, label=below:$\widehat{X}^n$] (t) at (-1.6,-4.4) {}
    edge [<-,thick] node[auto,swap] {$C_7$} node[auto] {$K$} (d)
    edge [<-,thick] node[auto] {} node[auto] {$C_4$} node[auto,swap] {$L$} (a);
	\node[dot, label=below:$\widehat{Y}^n$] (u) at (1.6,-4.4) {}
    edge [<-,thick] node[auto] {$C_8$} node[auto,swap] {$K$} (d)
    edge [<-,thick] node[auto,swap] {$C_5$} node[auto] {$R$} (b);
  \end{tikzpicture}
  \caption{Transmission of correlated sources over a butterfly network
    with secrecy constraint. A passive eavesdropper is on node
    $w$. $L$ and $R$ respectively are the messages on edge 4 and 5.}
  \label{fig:f2}
\end{figure}
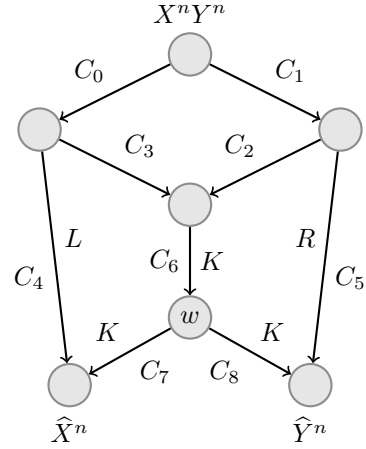
To see this, take a code of length $n$. Let $L$ and $R$ respectively denote the messages that are put on the edges with capacities $C_4$ and $C_5$. We have $nC_4\geq H(L)\geq I(L;X^n|K)\overset{(a)}{\cong}I(LK;X^n)\overset{(b)}{\cong} H(X^n)=nH(X)$. Approximation $(a)$ is a consequence of the fact that $K$ is almost independent of $X^n$, and $(b)$ follows from the fact that $X^n$ should (with high probability) be recoverable from $L$ and $K$. Therefore $C_4\geq H(X)$. Since $C_5+C_6=H(Y)$, that is the cut at the second sink is tight, both $K$ and $R$ must essentially be functions of $Y^n$. Thus we have $H(K)\cong I(K;Y^n|X^n)\leq H(Y^n|X^n)=nH(Y|X)$. Thus if $C_6>H(Y|X)$, the inequality $H(K)\leq nH(Y|X)$ implies that the edge with capacity $C_6$ is not fully used. But since $C_5+C_6=H(Y)$ and $Y^n$ is recoverable (with high probability) from $R$ and $K$, one must fully exploit the edge with capacity $C_6$. This is a contradiction.

These two examples can be recast in the same language if one considers the ``uncertainty" vector $[\frac{1}{n}H(K),~\frac{1}{n}H(K|X^n),~\frac{1}{n}H(K|Y^n),~ \frac{1}{n}H(K|X^n,Y^n)]$, {\it i.e.} the vector formed by listing the uncertainty left in $K$ conditioning on different subsets of $\{X^n, Y^n\}$. In the first example, each of $X^n$ and $Y^n$ is almost sufficient to determine $K$. Thus, the last three coordinates of the uncertainty vector are almost zero. Thus, the G\'{a}cs-K\"{o}rner common information can be reinterpreted as providing an upper bound for the first coordinate of the uncertainty vector when all the other coordinates are zero. In the second example, the secrecy constraint of $K$ being almost independent of $X^n$ imposes the constraint that the first and the second coordinate of the uncertainty vector are equal. The fact that $K$ is a function of $Y^n$ implies that the third and the fourth coordinate are almost zero. Thus the uncertainty vector is of the form $[a,~a,~0,~0]$. The constraint $C_6\leq H(Y|X)$ can be interpreted as saying that the maximum value of $a$ such that the uncertainty vector $[a,~a,~0,~0]$ is plausible, is $a=H(Y|X)$.

\section{The Uncertainty Region}
\label{Section:UncertaintyRegion}
The above section motivates the definition of the uncertainty region. In this section we formally define this region and then provide a complete characterization of it. In the next section we discuss the use of the uncertainty region in proving converses.

Given joint distribution $p(x,y)$ on discrete random variables $X$ and $Y$, let us define a four-dimensional region \emph{uncertainty region}, $U(p)$, as the closure of the set of non-negative 4-tuples $(u_1,u_2,u_3,u_4)$ such that for some $n$ and $p(k|x^n,y^n)$ we have
\begin{align*}
&u_1=\frac{1}{n}H(K),& u_2=\frac{1}{n}H(K|X^n),\\&u_3=\frac{1}{n}H(K|Y^n),& u_4=\frac{1}{n}H(K|X^n,Y^n).
\end{align*}
Intuitively speaking, the coordinates of this vector are the uncertainties of $K$ when i.i.d.\ copies of a subset of variables $X$ and $Y$ are available. We are interested in the set of all plausible uncertainty vectors. Note that we define the uncertainty region in terms of $p(x,y)$ alone, irrespective of the network architecture.

We now fully characterize the uncertainty region. The proof is provided in \cite{FullProofs}.
\begin{theorem}\label{Thm1} The region $U(p)$ is equal to the convex envelope of the union of the following four sets of points. The first set is the union over all $c\geq 0$ and $p(e|x,y)$ of non-negative 4-tuples $(u_1,u_2,u_3,u_4)$ where
\begin{align*}
&u_1=c+I(E;X,Y),\\&u_2=c+I(E;Y|X),\\&u_3=c+I(E;X|Y),\\&u_4=c.
\end{align*}
The second set of points is the union over all $c\geq 0$ of 4-tuples $(u_1,u_2,u_3,u_4)$ where
\begin{align*}
&u_1=c+H(Y|X),\\&u_2=c+H(Y|X),
\\&u_3=c,\\&u_4=c.
\end{align*}
The third set of points is the union over all $c\geq 0$ of 4-tuples $(u_1,u_2,u_3,u_4)$ where
\begin{align*}
&u_1=c+H(X|Y),\\&u_2=c,
\\&u_3=c+H(X|Y),\\&u_4=c.
\end{align*}
The fourth set of points is the union over all $c\geq 0$, $0\leq f\leq \max(H(X|Y),H(Y|X))$ of non-negative 4-tuples $(u_1,u_2,u_3,u_4)$ where
\begin{align*}
&u_1=c+f,\\&u_2=c+\min(f,H(Y|X)),
\\&u_3=c+\min(f,H(X|Y)),\\&u_4=c.
\end{align*}
\end{theorem}
\begin{remark}One can use the strengthened Carath\'{e}odory theorem of Fenchel \cite{Fenchel} to prove a cardinality bound of $|\mathcal{X}||\mathcal{Y}|+2$ on the auxiliary random variable $E$ in the first set of points.
\end{remark}

Although the above theorem characterizes the region, the following outer bound is useful in some instances. The extreme points of this outer bound belong to the first set of points of the above theorem.
\begin{theorem}\label{Thm2} The uncertainty region is a subset of the union over all $c,g,h\geq 0$ and $p(e|x,y)$ of 4-tuples $(u_1,u_2,u_3,u_4)$ where
\begin{align*}
&u_1=c+I(E;XY)\\&u_2=c+I(E;Y|X)+g\\&u_3=c+I(E;X|Y)+h\\&u_4=c.
\end{align*}
\end{theorem}

\section{Writing Converses Using the Uncertainty Region}
Take an arbitrary directed network $\mathcal{N}$ with a source $s$ and two sinks
$t_1$ and $t_2$. Suppose that the source observes i.i.d.\ copies of
$X$, $Y$ jointly distributed according to $p(x,y)$. Sink $t_1$ is interested in the i.i.d.\ copies of $X$, while
sink $t_2$ is interested in the i.i.d.\ copies of $Y$. The capacity of an edge $e$ is denoted by $C_e$. An $(n,\epsilon)$ code for this network consists of a set of encoding functions at the intermediate nodes such that $X^n$ and $Y^n$ can be recovered at the first and second sinks respectively with probabilities of error less than or equal to $\epsilon$, and furthermore the number of bits passed on a given edge $e$ is at most $n(C_e+\epsilon)$.

In order to write a converse for $\mathcal{N}$ we take the edges one by one and write a converse for that particular edge. At the end we intersect all such converses.

Take an $(n,\epsilon)$ code. Take a particular edge $e$ and let $K$ denote the random variable that is put on the edge $e$. The idea is to find as many constraints as possible on the uncertainty vector associated to $K$, {\it i.e.} $[\frac{1}{n}H(K), \frac{1}{n}H(K|X^n), \frac{1}{n}H(K|Y^n), \frac{1}{n}H(K|X^n,Y^n)]$. Let us denote the first coordinate $\frac{1}{n}H(K)$ by $d_e$, defined as the entropy rate of the random variable on edge $e$. This $d_e$ is required to satisfy $0\leq d_e\leq C_e+\epsilon$. Every cut that has the edge $e$ and separates the source from the first sink imposes a constraint on $\frac{1}{n}H(K|X^n)$ as follows.
\begin{lemma}\label{lemma1} Take an arbitrary cut (containing $e$) from the source to the first sink, and let $Cut_x$ denote the sum of the capacities of the edges on this cut. Then $\frac{1}{n}H(K|X^n)$ must satisfy the following inequality:
$$\frac{1}{n}H(K|X^n)\leq Cut_x-C_e+d_e-H(X)+k(\epsilon)$$
for some function $k(\epsilon)$ that converges to zero as $\epsilon$ converges to zero.\end{lemma}

\begin{proof} Let $Q$ denote the collection of random variables passing over the edges of the cut (except $e$). As shown in \cite{FullProofs}, $\frac{1}{n}H(Q)\leq Cut_x-C_e+m\epsilon$, where $m$ is the number of edges in the graph. Since $(Q,K)$ is the  collection of the random variables passing the edges of the cut, $X^n$ should be recoverable from $(Q,K)$ with probability of error less than or equal to $\epsilon$. Thus, by Fano's inequality $\frac{1}{n}H(X^n|Q,K)\leq k_1(\epsilon)$ for some function $k_1(\epsilon)$ that converges to zero as $\epsilon$ converges to zero. We have
\begin{align*}\frac{1}{n}H(K|&X^n)\leq \frac{1}{n}H(Q,K|X^n)=\frac{1}{n}H(Q,K,X^n)-\frac{1}{n}H(X^n)\\&\leq\frac{1}{n}H(Q)+\frac{1}{n}H(K)+\frac{1}{n}H(X^n|Q,K)-H(X)\\&\leq Cut_x-C_e+m\epsilon+d_e-H(X)+k_1(\epsilon).\end{align*}
We get the desired result by setting $k(\epsilon)=k_1(\epsilon)+m\epsilon$.
\end{proof}
Other restrictions on $\frac{1}{n}H(K|X^n)$ may come from secrecy constraints. For instance if $K$ is observed by an eavesdropper and there is an equivocation rate constraint on how much the eavesdropper can learn about $X^n$, say $\frac{1}{n}I(K;X^n)\leq R$, we can conclude that $\frac{1}{n}H(K|X^n)\geq \frac{1}{n}H(K)-R=d_e-R$.

One can use similar ideas to impose constraints on $\frac{1}{n}H(K|Y^n)$.

If there is no secrecy constraint, without loss of generality we assume that $K$ is a function of $(X^n,Y^n)$ as randomized coding would only reduce the throughput. Thus the last coordinate $\frac{1}{n}H(K|X^n,Y^n)$ will be zero. The following lemma (whose proof is similar to that of Lemma \ref{lemma1}, and hence is omitted) is also useful.
\begin{lemma}\label{lemma2} Take an arbitrary cut containing $e$ from the source to the first sink, and let $Cut_{x,y}$ denote the sum of the capacities of the edges on this cut. Then $\frac{1}{n}H(K|X^n,Y^n)$ must satisfy the following inequality:
$$\frac{1}{n}H(K|X^n,Y^n)\leq Cut_{x,y}-C_e+d_e-H(X,Y)+k(\epsilon)$$
for some function $k(\epsilon)$ that converges to zero as $\epsilon$ converges to zero.\end{lemma}

Thus for every $(n,\epsilon)$ code we write all such constraints on the coordinates of $$\left[\frac{1}{n}H(K), \frac{1}{n}H(K|X^n), \frac{1}{n}H(K|Y^n), \frac{1}{n}H(K|X^n,Y^n)\right].$$ Lastly we look at these constraints over a sequence of codes $(n_i,\epsilon_i)$ where $\epsilon_i\rightarrow 0$ as $i\rightarrow \infty$. As an example, consider a problem with no secrecy constraints. Let $Mincut^e_x$ be the smallest cut that has the edge $e$ and separates the source from the first sink. $Mincut^e_{y}$ and $Mincut^e_{x,y}$ are defined similarly. For the code $(n_i,\epsilon_i)$ we have
\begin{align*}\frac{1}{n_i}H(K_i)=&d_{ei},\\
\frac{1}{n_i}H(K_i|X^{n_i})\leq& Mincut^e_x-C_e\\&+d_{ei}-H(X)+k(\epsilon_i),\\
\frac{1}{n_i}H(K_i|Y^{n_i})\leq& Mincut^e_y-C_e\\&+d_{ei}-H(Y)+k(\epsilon_i),\\
\frac{1}{n_i}H(K_i|X^{n_i},Y^{n_i})=0\leq& Mincut^e_{x,y}-C_e\\&+d_{ei}-H(X,Y)+k(\epsilon_i).
\end{align*}
There is a convergent subsequence $d_{ei}$ converging to some $d_e^*\leq C_e$. Therefore the region $U(p)$ contains a point $[u_1, u_2, u_3, u_4]$ such that
\begin{align*}&u_1=d_{e}^*,\\&
u_2\leq Mincut^e_x-C_e+d_{e}^*-H(X),\\&
u_3\leq Mincut^e_y-C_e+d_{e}^*-H(Y),\\&
u_4=0\leq Mincut^e_{x,y}-C_e+d_{e}^*-H(X,Y).
\end{align*}
From Theorem 2 we know that there exist $c,g,h\geq 0$ and $p(e|x,y)$ such that
\begin{align*}
&u_1=c+I(E;X,Y),~~~~~~~~u_2=c+I(E;Y|X)+g,\\&u_3=c+I(E;X|Y)+h,~~~~u_4=c.
\end{align*}
Thus, there exists a $p(e|x,y)$ such that
\begin{align}&d_{e}^*=I(E;X,Y)\leq C_e\label{eqn:edgeconverse1}\\&
Mincut^e_x-C_e+d_{e}^*-H(X)\geq I(E;Y|X)\label{eqn:edgeconverse2}\\&
Mincut^e_y-C_e+d_{e}^*-H(Y)\geq I(E;X|Y).\label{eqn:edgeconverse3}
\end{align}
And furthermore $0\leq Mincut^e_{x,y}-C_e+d_{e}^*-H(X,Y)$. These inequalities together form a converse for the edge $e$. We can repeat this process for all the edges and take intersection over all such converses.

\subsection{Comparison with the cut-set bound}
Let us compare the above converse with the one given by the cut-set bound. Take some edge $e$. The constraints
\begin{align*}&d_{e}^*=I(E;X,Y)\leq C_e,\\&
Mincut^e_x-C_e+d_{e}^*-H(X)\geq I(E;Y|X),\\&
Mincut^e_y-C_e+d_{e}^*-H(Y)\geq I(E;X|Y),\\&
Mincut^e_{x,y}-C_e+d_{e}^*-H(X,Y)\geq 0
\end{align*}
imply that $Mincut^e_x-H(X)\geq 0$, $Mincut^e_y-H(Y)\geq 0$ and $Mincut^e_{x,y}-H(X,Y)\geq 0$. Since edge $e$ was arbitrary, one can see that this converse is no worse than the cut-set bound. Let us consider the network given in figure \ref{fig:f3}. Assume that $C_3=C_4=C_5$. This network is known as the Gray-Wyner system \cite{GrayWyner}.
 \begin{figure}
   \centering
    \begin{tikzpicture}[scale=1.5]
     \node[dot, label=above:$X^nY^n$] (c) at (0,0) {};
     \node[dot] (d) at (0,-1.3) {}
	edge [<-,thick] node[auto, pos=0.3] {$C_3$} (c);
     \node[dot, label=below:$\widehat{X}^n$] (t) at (-1.2,-2) {}
	edge [<-,thick] node[auto, swap] {$C_4$} (d)
	edge [<-,thick] node[auto] {$C_1$} (c);
     \node[dot, label=below:$\widehat{Y}^n$] (u) at (1.2,-2) {}
	edge [<-,thick] node[auto] {$C_5$} (d)
	edge [<-,thick] node[auto, swap] {$C_2$} (c);
   \end{tikzpicture}
   \caption{This network is the Gray-Wyner system when $C_3=C_4=C_5$.}
   \label{fig:f3}
 \end{figure}
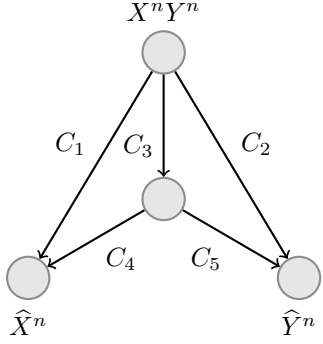
 Let us write the converse for the edge number 3. The converse says that there exists a $p(e|x,y)$ such that
\begin{align*}&d_{3}^*=I(E;X,Y)\leq C_3,\\&
Mincut^3_x-C_3+d_{3}^*-H(X)\geq I(E;Y|X),\\&
Mincut^3_y-C_3+d_{3}^*-H(Y)\geq I(E;X|Y),\\&
Mincut^3_{x,y}-C_3+d_{3}^*-H(X,Y)\geq 0.
\end{align*}
Note that $Mincut^3_x=C_4+C_1=C_3+C_1$, $Mincut^3_y=C_5+C_2=C_3+C_2$ and $Mincut^3_{x,y}=C_1+C_2+C_3$. Thus
\begin{align*}&d_{3}^*=I(E;X,Y)\leq C_3,\\&
C_3+C_1-C_3+d_{3}^*-H(X)\geq I(E;Y|X),\\&
C_3+C_2-C_3+d_{3}^*-H(Y)\geq I(E;X|Y),\\&
C_1+C_2+C_3-C_3+d_{3}^*-H(X,Y)\geq 0.
\end{align*}
After simplification and substituting the value of $d_{3}^*=I(E;X,Y)$ from the first equation into the other equations we get that
\begin{align*}&C_3\geq I(E;X,Y),\\&
C_1\geq I(E;Y|X)-I(E;X,Y)+H(X)=H(X|E),\\&
C_2\geq I(E;X|Y)-I(E;X,Y)+H(Y)=H(Y|E),\\&
C_1+C_2\geq H(X,Y)-I(E;X,Y)=H(X,Y|E).
\end{align*}
The last equation is redundant. Therefore we get
\begin{align*}&C_3\geq I(E;X,Y), C_1\geq H(X|E), C_2\geq H(Y|E)
\end{align*}
for some $p(e|x,y)$. But this is exactly the solution to the Gray-Wyner system \cite{GrayWyner}. Therefore the new converse is tight. On the other hand the cut-set bound is not tight for this network. Let us consider the minimum of $C_3$ such that $C_1+C_2+C_3=H(X,Y)$ over the actual region and the cut-set bound. It is known that in the Gray-Wyner system this minimum is equal to the Wyner's common information \cite{Wyner}. However, in the cut-set bound this minimum is $I(X;Y)$ which can be strictly less than the Wyner's common information. Therefore the new converse represents a strict improvement over the cut-set bound.

\subsection{Using ``Edge-Cuts" to write better converses}
The new converse as expressed above is not also tight in general. In the above discussion we observed that every cut that has the edge $e$ and separates the source from the first sink imposes a constraint on $\frac{1}{n}H(K|X^n)$. However it turns out that one can use the technique to write strictly better converses by looking at what might be termed ``edge-cuts" (certain cuts in certain subgraphs of the original graph) if there are multiple source nodes in the network. Our concept of edge-cuts should not be confused with that of \cite{KramerSavari}. 

In order to construct an explicit example for multi-source problems that shows the benefit of using edge-cuts, we consider a directed network with two sources $s_1$ and $s_2$ and two sinks
$t_1$ and $t_2$ of Figure \ref{fig:f3N} under the assumption that $C_6=C_7=C_8$.
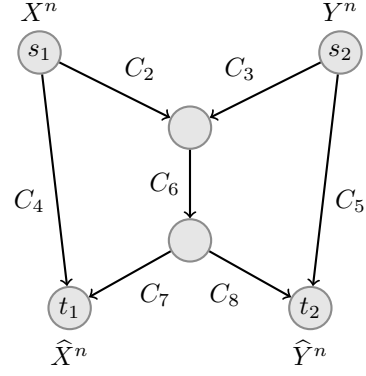
\begin{figure}
\centering
  \begin{tikzpicture}[scale=1]
	\node[dot, label=above:$X^n$] (a) at (-2,-1) {$s_1$};
	\node[dot, label=above:$Y^n$] (b) at (2,-1) {$s_2$};
	\node[dot] (c) at (0,-2) {}
    edge [<-,thick] node[auto] {$C_3$} (b)
    edge [<-,thick] node[auto,swap] {$C_2$} (a);
	\node[dot] (d) at (0,-3.5) {}
    edge [<-,thick] node[auto] {$C_6$} (c);
	\node[dot, label=below:$\widehat X^n$] (t) at (-1.6,-4.4) {$t_1$}
    edge [<-,thick] node[auto,swap] {$C_7$} (d)
    edge [<-,thick] node[auto] {} node[auto] {$C_4$} (a);
	\node[dot, label=below:$\widehat Y^n$] (u) at (1.6,-4.4) {$t_2$}
    edge [<-,thick] node[auto] {$C_8$} (d)
    edge [<-,thick] node[auto,swap] {$C_5$} (b);
  \end{tikzpicture}
\caption{An explicit example for a multi-source problem that shows the benefit of using edge-cuts. We write the edge-cut for edge 6.}\label{fig:f3N}
\end{figure}

 Suppose that the source $s_1$ observes i.i.d.\ copies of the random variable
$X$, and source $s_2$ observes i.i.d.\ copies of the random variable $Y$. As before, random variables $X$ and $Y$ are jointly distributed according to $p(x,y)$, and sink $t_1$ is interested in the i.i.d. copies of $X$ while
sink $t_2$ is interested in the i.i.d. copies of $Y$. We consider the problem of
reliable transmission to fulfill the demands of both
sink nodes, with probability of decoding error converging to zero as the number of i.i.d. observations of $X$, $Y$ grows without bound.

\subsubsection{edge-cuts}
Take an arbitrary edge $e$ in a directed graph from a vertex $v_1$ to a vertex $v_2$. Consider the subgraph formed by including all the directed paths from the two sources to $v_2$. We can think of $v_2$ as an imaginary sink in this subgraph. Let $K$ denote the random variable carried on the $v_1-v_2$ edge. We can consider three types of cuts between the two sources and the imaginary sink in this subgraph: 1. cuts that that separate the first source from node $v_2$ but do not separate the second source from node $v_2$, 2: cuts that separate the second source from $v_2$ but do not separate the first source from node $v_2$, and 3. cuts that separate both sources from node $v_2$. Let $Cut_{x,y,v_2}$ denote the sum-capacity of an arbitrary cut that separates both sources from node $v_2$ in the subgraph. We have
$$Cut_{x,y,v_2}\geq \frac{1}{n}I(K;X^n,Y^n)$$
Let $Cut_{x,v_2}$ denote the sum-capacity of an arbitrary cut that separates the first source from node $v_2$ in the subgraph. We have
$$Cut_{x,v_2}\geq \frac{1}{n}I(K;X^n|Y^n)$$
Similarly, let $Cut_{y,v_2}$ denote the sum-capacity of an arbitrary cut that separates the second source from node $v_2$ in the subgraph. We have
$$Cut_{y,v_2}\geq \frac{1}{n}I(K;Y^n|X^n)$$
These inequalities have consequences for the uncertainty vector $[\frac{1}{n}H(K),~\frac{1}{n}H(K|X^n),~\frac{1}{n}H(K|Y^n),~ \frac{1}{n}H(K|X^n,Y^n)]$.

Consider the edge $6$ in Figure \ref{fig:f3N}. The resulting subgraph formed by including all the directed paths from the two sources to the end point of this edge is shown in Figure \ref{fig:f4}.
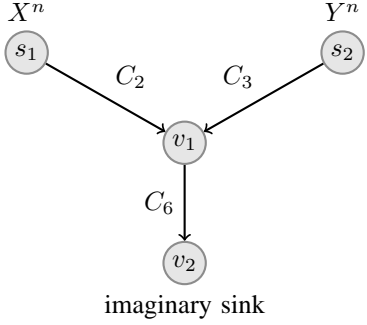
\begin{figure}
\centering
  \begin{tikzpicture}[scale=1]
     \node[dot, label=below:imaginary sink] (c) at (0,-1.6) {$v_2$};
     \node[dot] (d) at (0,0) {$v_1$}
	edge [->,thick] node[auto,swap] {$C_6$} (c);
     \node[dot, label=above:$X^n$] (t) at (-2.1,1.2) {$s_1$}
	edge [->,thick] node[auto] {$C_2$} (d);
     \node[dot, label=above:$Y^n$] (u) at (2.1,1.2) {$s_2$}
	edge [->,thick] node[auto,swap] {$C_3$} (d);
   \end{tikzpicture}
\caption{The subgraph formed by including all the directed paths from the two sources to the end point of edge $6$, i.e. the node $v_2$. We can think of $v_2$ as an imaginary sink in this subgraph. Edge-cuts are the cuts between the two sources and the imaginary sink in this subgraph.}\label{fig:f4}
\end{figure}
Let $K_6$ denote the random variable carried on this edge. Observe that edge $2$ is a cut that separates the first source only from the imaginary sink. Therefore we can write $\frac{1}{n}I(K_6;X^n|Y^n)\leq C_2$. Since $H(K_6|X^n,Y^n)=0$, we conclude that $\frac{1}{n}H(K_6|Y^n)\leq C_2$. It is not possible to get this constraint on the uncertainty of $K_6$ given $Y^n$ by looking at the cuts between the sources and the sinks in the original graph. To see this note that if we use equations (\ref{eqn:edgeconverse1}-\ref{eqn:edgeconverse3}) for all the cuts that have the edge $6$ we get the following set of equations:
\begin{align*}
&d_6=I(E_6;XY)\leq C_6\\
&C_4+C_6-C_6+d_6-H(X)\geq I(E_6;Y|X)\\&~~~~\mbox{because }\{4,7\}\mbox{ is a cut between }s_1,s_2\\&~~~~\mbox{and }t_1\mbox{ in the original graph}\\
&C_5+C_6-C_6+d_6-H(Y)\geq I(E_6;X|Y)\\&~~~~\mbox{because }\{5,8\}\mbox{ is a cut between }s_1,s_2\\&~~~~\mbox{and }t_2\mbox{ in the original graph}\\
\end{align*}
for some $p(e_6|x,y)$. Here we used the fact that the capacities of edges 6, 7 and 8 are all the same, hence we can assume that they are all carrying the same message. Therefore we can compute the uncertainty of the message on edge 6 by looking at cuts that include edge 7 or 8.

The next step is to incorporate the inequality $\frac{1}{n}H(K_6|Y^n)\leq C_2$ with the above set of inequalities. Remember that $C_5+C_6-C_6+d_6-H(Y)$ in the third inequality above is an upper bound on $\frac{1}{n}H(K_6|Y^n)$. This comes from Lemma \ref{lemma1}. The term $I(E_6;X|Y)$ is a lower bound on $\frac{1}{n}H(K_6|Y^n)$. This comes from Theorem \ref{Thm2}. Now, using the inequality $\frac{1}{n}H(K_6|Y^n)\leq C_2$ we can conclude that $\min\big(C_2, C_5+C_6-C_6+d_6-H(Y)\big)$ is an upper bound on $\frac{1}{n}H(K_6|Y^n)$. Thus, we can write
\begin{align*}
&d_6=I(E_6;XY)\leq C_6\\
&C_4+C_6-C_6+d_6-H(X)\geq I(E_6;Y|X)\\&~~~~\mbox{because }\{4,7\}\mbox{ is a cut between }s_1,s_2\\&~~~~\mbox{and }t_1\mbox{ in the original graph}\\
&\min\big(C_2, C_5+C_6-C_6+d_6-H(Y)\big)\geq I(E_6;X|Y)\\&~~~~\mbox{because }\{5,8\}\mbox{ is a cut between }s_1,s_2\\&~~~~\mbox{and }t_2\mbox{ in the original graph}\\
\end{align*}
for some $p(e_6|x,y)$. This set of equations can be simplified in the following form
\begin{align}
&C_6\geq I(E_6;XY)\label{eqn:EdgeCutFirst}\\
&C_4\geq H(X|E_6)\label{eqn:EdgeCutSecond}\\
&C_5\geq H(Y|E_6)\\
&C_2\geq I(E_6;X|Y)\label{eqn:EdgeCutLast}
\end{align}
for some $p(e_6|x,y)$.
\subsubsection{Comparison of two converses}
We now compare the converse given by equations (\ref{eqn:converseeq-first}- \ref{eqn:converseeq-last}) with the converse given by equations (\ref{eqn:EdgeCutFirst}-\ref{eqn:EdgeCutLast}). The latter converse is derived in the appendix by looking at all cuts between the sources and the sinks (no edge-cuts here).

We claim that the minimum possible value of $C_6$ in this converse is less than or equal to $I(X;Y)$ if we restrict ourselves to networks where $C_2+C_4=H(X|Y)$. This is shown at the end of the appendix. Next consider the converse written using edge-cuts and given by equations (\ref{eqn:EdgeCutFirst}-\ref{eqn:EdgeCutLast}). We show that the minimum in the other converse is $\min_{X\rightarrow E\rightarrow Y}I(E;XY)$, i.e. Wyner's common information. From equations \ref{eqn:EdgeCutSecond} and \ref{eqn:EdgeCutLast} we have $C_2+C_4\geq H(X|E_6)+I(E_6;X|Y)=H(X|E_6)+H(X|Y)-H(X|E_6,Y)=H(X|Y)+I(X;Y|E_6)$. If we restrict ourselves to networks where $C_2+C_4=H(X|Y)$, it must be the case that random variables $X\rightarrow E_6\rightarrow Y$ form a Markov chain. Therefore the minimum of $C_6$ is $\min_{X\rightarrow E_6\rightarrow Y}I(E_6;X,Y)$ which is equal to Wyner's common information.

Noting that Wyner's common information is in general larger than $I(X;Y)$, we conclude that the later converse is strictly better than the former converse.

\section{Proofs}
\label{Section:Proofs}
\begin{proof}[Proof of Theorem \ref{Thm1}]:

\emph{Achievability}:
We begin by showing that each of the four set of points is a subset of $U(p)$. This would complete the proof noting that $U(p)$ is a convex set in $\mathbb{R}^4$ as it implies that the convex envelope of the union of the four sets of points is also a subset of $U(p)$. The details of $U(p)$ being a convex set are given in~\cite{FullProofs}. Note that if we can prove the inclusion for $c=0$ in each case, we will have it for all $c\geq 0$ since we can always add noise to $K$ that is independent of all previously defined random variables. Let us begin with the first set of points. Take some arbitrary $p(e|x,y)$. We would like to find a sequence of $p(k_n,x^n,y^n)$ such that
\begin{align*}
  & \lim_{n\rightarrow\infty}\frac{1}{n}H(K_n)=I(E;X,Y)\\
  &  \lim_{n\rightarrow\infty}\frac{1}{n}H(K_n|X^n)=I(E;Y|X)\\
   & \lim_{n\rightarrow\infty}\frac{1}{n}H(K_n|Y^n)=I(E;X|Y)\\
   & \lim_{n\rightarrow\infty}\frac{1}{n}H(K_n|X^n,Y^n)=0\\
\end{align*}
We use part 1 of Theorem 5 of \cite{NewShannon} which says that one can find a sequence of $p(k_n,x^n,y^n)$ such that
\begin{align*}
  & \lim_{n\rightarrow\infty}\frac{1}{n}I(X^n;Y^n|K_n)=I(X;Y|E)\\
  &  \lim_{n\rightarrow\infty}\frac{1}{n}H(K_n|X^n)=I(E;Y|X)\\
   & \lim_{n\rightarrow\infty}\frac{1}{n}H(K_n|Y^n)=I(E;X|Y)\\
   & \lim_{n\rightarrow\infty}\frac{1}{n}H(K_n|X^n,Y^n)=0\\
\end{align*}
The difference between these set of equations and the ones we would like to have is the first one. However these four set of equations are indeed equivalent. Note that
\begin{align*}
    H(K_n)=&H(K_n|X^n)+H(K_n|Y^n)\\&-H(K_n|X^n,Y^n)+I(X^n;Y^n)-I(X^n;Y^n|K_n).
\end{align*}
Thus,
\begin{align*}
    \lim_{n\rightarrow\infty}\frac{1}{n}H(K_n)&= \lim_{n\rightarrow\infty}\frac{1}{n}H(K_n|X^n)+\lim_{n\rightarrow\infty}\frac{1}{n}H(K_n|Y^n)\\&~~~-\lim_{n\rightarrow\infty}\frac{1}{n}H(K_n|X^n,Y^n)+I(X;Y)\\&~~~-\lim_{n\rightarrow\infty}\frac{1}{n}I(X^n;Y^n|K_n)\\&=
    I(E;Y|X)+I(E;X|Y)\\&~~~+I(X;Y)-I(X;Y|E)\\&=I(E;X,Y).
\end{align*}

We now prove that the second and the third sets of points is in $U(p)$. Slepian-Wolf tell us that for any $\epsilon$ one can find $N$ such that for any $n>N$ there are functions $M_{xn}:\mathcal{X}^n\mapsto [1:2^{n(H(X|Y)+\epsilon)}]$ and $M_{yn}:\mathcal{Y}^n\mapsto [1:2^{n(H(Y|X)+\epsilon)}]$ such that $X^n$ can be recovered from $(M_{xn}(X^n),Y^n)$, and $Y^n$ can be recovered from $(M_{yn}(Y^n),X^n)$ with probability $1-\epsilon$. One can prove that\footnote{For instance the first equation holds because $\frac{1}{n}I(M_{xn}(X^n);Y^n)=\frac{1}{n}(H(M_{xn}(X^n))+H(Y^n)-H(M_{xn}(X^n),Y^n))=
\frac{1}{n}(H(M_{xn}(X^n))+H(Y^n)-H(X^n,Y^n)+H(X^n|M_{xn}(X^n),Y^n))\leq H(X|Y)+\epsilon+H(Y)-H(X,Y)+h(\epsilon)+\epsilon|\mathcal{X}||\mathcal{Y}|$ by the Fano inequality and the fact that $M_{xn}$ is a function of $X^n$. The third equation holds because it is possible to reconstruct $(X^n, Y^n)$ from $M_{xn}(X^n)$ and $Y^n$ with high probability.}
\begin{align}
  & \frac{1}{n}I(M_{xn}(X^n);Y^n)\leq r_1(\epsilon),\label{eqn:P1}\\
  & \frac{1}{n}I(M_{yn}(Y^n);X^n)\leq r_2(\epsilon),\label{eqn:P2}\\
  & \frac{1}{n}H(M_{xn}(X^n))\geq H(X|Y)-r_3(\epsilon),\label{eqn:P3}\\
  & \frac{1}{n}H(M_{yn}(Y^n))\geq H(Y|X)-r_4(\epsilon).\label{eqn:P4}
\end{align}
for some functions $r_i$ such that $r_i(\epsilon)$ converges to zero as $\epsilon$ converges to zero. Setting $K_n=M_{yn}(Y^n)$ would give us the second set of points as $\epsilon\rightarrow 0$ and $n\rightarrow \infty$. To see this note that $\lim_{n\rightarrow\infty}\frac{1}{n}H(K_n)=H(Y|X)$ because of equation (\ref{eqn:P4}) and the fact that $M_{yn}$ is taking value in $[1:2^{n(H(Y|X)+\epsilon)}]$. Furthermore one can show that $\lim_{n\rightarrow\infty}\frac{1}{n}H(K_n|X^n)=H(Y|X)$ using equation (\ref{eqn:P2}). Similarly setting $K_n=M_{xn}(X^n)$ asymptotically gives us the third set of points.

We now prove that the fourth set of points is in $U(p)$. In order to define $K_n$ appropriately to get this set of points we are going to use random variables $M_{yn}$ and $M_{xn}$ defined above. For every $n\in \mathbb{N}$, we can find some $\epsilon_n$ such that equations \ref{eqn:P1}-\ref{eqn:P4} hold, and that $\epsilon_n$ converges to zero as $n$ converges to infinity. Next, take some arbitrary $0\leq f\leq \max(H(X|Y),H(Y|X))$. We would like to find a sequence of $p(k_n,x^n,y^n)$ such that
\begin{align*}
  & \lim_{n\rightarrow\infty}\frac{1}{n}H(K_n)=f\\
  &  \lim_{n\rightarrow\infty}\frac{1}{n}H(K_n|X^n)=\min(f,H(Y|X))\\
   & \lim_{n\rightarrow\infty}\frac{1}{n}H(K_n|Y^n)=\min(f,H(X|Y))\\
   & \lim_{n\rightarrow\infty}\frac{1}{n}H(K_n|X^n,Y^n)=0.\\
\end{align*}

Let us define the functions $M_{xn}\in [1:2^{n(H(X|Y)+\epsilon_n)}]$ and $M_{yn}\in [1:2^{n(H(Y|X)+\epsilon_n)}]$ as above. We can think of $M_{xn}(X^n)$ and $M_{yn}(Y^n)$ as two random binary sequences of length $\lfloor n(H(X|Y)+\epsilon_n)\rfloor$ and $\lfloor n(H(Y|X)+\epsilon_n)\rfloor$ respectively. Let us use the notation $M_{yn}^{i:j}(Y^n)$ to denote the set of $i^{th}$ to $j^{th}$ bits of $M_{yn}(Y^n)$. We use a similar notation for $M_{xn}(X^n)$.

Without loss of generality let us assume that $H(X|Y)\geq H(Y|X)$. Consider the following two cases:

\emph{Case 1.} $f\leq H(Y|X)$:

 In this case, we let $K_n$ be equal to the bitwise XOR of the first $\lfloor nf\rfloor$ bits of $M_{xn}(X^n)$ and $M_{yn}(Y^n)$, i.e. the bitwise XOR of $M_{xn}^{1:\lfloor nf\rfloor}(X^n)$ and $M_{yn}^{1:\lfloor nf\rfloor}(Y^n)$. Clearly $\frac{1}{n}H(K_n|X^n,Y^n)=0$. We would like to show that
 \begin{align*}
  & \lim_{n\rightarrow\infty}\frac{1}{n}H(K_n)=f,\\
  & \lim_{n\rightarrow\infty}\frac{1}{n}H(K_n|X^n)=f,\\
  & \lim_{n\rightarrow\infty}\frac{1}{n}H(K_n|Y^n)=f.\\
\end{align*}
It suffices to prove the last two inequalities since $H(K_n|X^n)\leq H(K_n)\leq \log|\mathcal{K}_n|\leq nf$. We prove the second one, the proof for the third is similar. Note that $H(K_n|X^n)=H(K_n|X^n,M_{xn}^{1:\lfloor nf\rfloor}(X^n))=H(M_{yn}^{1:\lfloor nf\rfloor}(Y^n)|X^n)$. Equation \ref{eqn:P2} implies that
 \begin{align*}
  & \frac{1}{n}I(M_{yn}^{1:\lfloor nf\rfloor}(Y^n);X^n)\leq r_2(\epsilon_n).
\end{align*}
Thus,
 \begin{align*}
  & \lim_{n\rightarrow\infty}\frac{1}{n}H(K_n|X^n)=\lim_{n\rightarrow\infty}\frac{1}{n}H(M_{yn}^{1:\lfloor nf\rfloor}(Y^n)).\\
\end{align*}
Clearly $\lim_{n\rightarrow\infty}\frac{1}{n}H(M_{yn}^{1:\lfloor nf\rfloor}(Y^n))\leq f$. If $\lim_{n\rightarrow\infty}\frac{1}{n}H(M_{yn}^{1:\lfloor nf\rfloor}(Y^n))<f$ then additionally considering the $\lfloor nf\rfloor+1$ to $\lfloor nH(Y|X)+n\epsilon_n\rfloor$ bits of $M_{yn}$ can at most increase the asymptotic entropy rate by $H(Y|X)-f$ bits. On the other hand equation \ref{eqn:P4} implies that $\lim_{n\rightarrow\infty}\frac{1}{n}H(M_{yn}(Y^n))=H(Y|X)$. This is a contradiction because using the fact that the joint entropy is less than or equal to the individual entropies one can write
\small \begin{align*}\lim_{n\rightarrow\infty}\frac{1}{n}H(M_{yn}(Y^n))&\leq \lim_{n\rightarrow\infty}\frac{1}{n}H(M_{yn}^{1:\lfloor nf\rfloor}(Y^n))\\&~~~+\lim_{n\rightarrow\infty}\frac{1}{n}H(M_{yn}^{\lfloor nf\rfloor+1:\lfloor nH(Y|X)+n\epsilon_n\rfloor}(Y^n))\\&<f+n-f=n.\end{align*}
\normalsize
\emph{Case 2.} $H(Y|X)\leq f\leq H(X|Y)$: In this case, let $K_n$ be equal to the bitwise XOR of $M_{xn}^{1:\lfloor nH(Y|X)\rfloor}(X^n)$ and $M_{yn}^{1:\lfloor nH(Y|X)\rfloor}(Y^n)$, together with $M_{xn}^{\lfloor nH(Y|X)\rfloor+1:\lfloor nf\rfloor}(X^n)$. In this case, one needs to show that
 \begin{align*}
  & \lim_{n\rightarrow\infty}\frac{1}{n}H(K_n)=f,\\
  & \lim_{n\rightarrow\infty}\frac{1}{n}H(K_n|X^n)=nH(Y|X),\\
  & \lim_{n\rightarrow\infty}\frac{1}{n}H(K_n|Y^n)=f.\\
\end{align*}
As in case 1, the third equation implies the first. The proof for the last two limits is similar to the one discussed above in case 1.

\emph{Converse}:
Since $U(p)$ is convex, to show that the region $U(p)$ is equal to the convex envelope of the given set of points, it suffices to show that for any real $\lambda_1$, ..., $\lambda_4$, the maximum of $\lambda_1 u_1+\lambda_2 u_2 + \lambda_3 u_3+\lambda_4 u_4$ over $U(p)$ is achieved at one of the given points. We show this by a case by case analysis. First assume that $\lambda_1+\lambda_2+\lambda_3+\lambda_4>0$. In this case maximum will be infinity and is achieved at the point $[c,c,c,c]$ when $c\rightarrow \infty$. If $\lambda_1+\lambda_2+\lambda_3+\lambda_4\leq 0$, we can write the maximum of $\lambda_1 u_1+\lambda_2 u_2 + \lambda_3 u_3+\lambda_4 u_4$ over $U(p)$ as
 \begin{align*}&\limsup_{n\rightarrow \infty}\frac{1}{n}\bigg(\lambda_1I(K;X^nY^n)+\lambda_2I(K;Y^n|X^n)+\\&\lambda_3I(K;X^n|Y^n)+(\lambda_1+\lambda_2+\lambda_3+\lambda_4)H(K|X^n,Y^n)\bigg).\end{align*}
The last term $(\lambda_1+\lambda_2+\lambda_3+\lambda_4)H(K|X^n,Y^n)$ is less than or equal to zero. Given any $(K,X^n,Y^n)$, we can always use part 1 of Theorem 5 of \cite{NewShannon} as in the achievability to find $(K',X^{nm}, Y^{nm})$ for some $m$ such that $K'$ is a function of $(X^{nm}, Y^{nm})$ and sum of the first three terms is asymptotically unchanged. $K'$ being a function of $(X^{nm}, Y^{nm})$ implies that $(\lambda_1+\lambda_2+\lambda_3+\lambda_4)H(K'|X^{nm},Y^{nm})$ is zero. To sum up, without loss of generality we can consider only random variables $K$ that are deterministic functions of $(X^n,Y^n)$, and furthermore we only need to compute the following expression over such random variables
 \begin{align*}&\limsup_{n\rightarrow \infty}\frac{1}{n}\bigg(\lambda_1I(K;X^nY^n)+\lambda_2I(K;Y^n|X^n)\\&+\lambda_3I(K;X^n|Y^n)\bigg).\end{align*}
We now continue by a case by case analysis:
\begin{itemize}
  \item $\lambda_1\geq 0,\lambda_2\geq 0,\lambda_3\geq 0$: Note that if we replace $K$ with $(K,X^n,Y^n)$ the expression will not decrease. Since $K$ is a function of $(X^n,Y^n)$, we conclude that $K=X^nY^n$ is the optimal choice in this instance. In this case the maximum of $\lambda_1 u_1+\lambda_2 u_2 + \lambda_3 u_3+\lambda_4 u_4$ over $U(p)$ will be equal to the maximum of the same expression over the first set of points with the choice of $E=XY$.
  \item $\lambda_1\geq 0,\lambda_2\leq 0,\lambda_3\geq 0$: If $\lambda_1+\lambda_2\geq 0$, the maximum of $\lambda_1 u_1+\lambda_2 u_2 + \lambda_3 u_3+\lambda_4 u_4$ over $U(p)$ will be equal to the maximum of the same expression over the first set of points with the choice of $E=XY$. To see this write $\lambda_2I(K;Y^n|X^n)$ as $\lambda_2I(K;Y^n,X^n)-\lambda_2I(K;X^n)$ and note that the expression is maximized when $K=X^nY^n$. If $\lambda_1+\lambda_2\leq 0$ first note that if we replace  $K$ with $(K,X^n)$ the expression will not decrease. In this case the expression $\lambda_1I(K,X^n;X^nY^n)+\lambda_2I(K,X^n;Y^n|X^n)+\lambda_3I(K,X^n;X^n|Y^n)$ will be equal to
      $\lambda_1H(X^n)+\lambda_3H(X^n|Y^n)+(\lambda_1+\lambda_2)I(K;Y^n|X^n)$. Since $\lambda_1+\lambda_2\leq 0$, we have $(\lambda_1+\lambda_2)I(K;Y^n|X^n)\leq 0$. Thus the maximum of $\lambda_1 u_1+\lambda_2 u_2 + \lambda_3 u_3+\lambda_4 u_4$ over $U(p)$ will be less than or equal to $\lambda_1H(X)+\lambda_3H(X|Y)$, which is equal to the maximum of the same expression over the first set of points with the choice of $E=X$.
  \item $\lambda_1\geq 0,\lambda_2\geq 0,\lambda_3\leq 0:$ This case is similar to case 2 by symmetry.
  \item $\lambda_1\geq 0, \lambda_2\leq 0,\lambda_3\leq 0:$ Take some arbitrary $n$ and $K=f(X^n,Y^n)$. Let the random index $J$ be uniformly distributed on $\{1,2,3,...,n\}$ and independent of $(K,X^n,Y^n)$. Define the auxiliary random variables $E=(K,X_{1:J-1},Y_{1:J-1},J), X=X_J, Y=Y_J$. Note that
      \begin{align*}I(K;X^n,Y^n)=&\sum_{j=1}^nI(K;X_j,Y_j|X_{1:j-1},Y_{1:j-1})\\&=\sum_{j=1}^nI(K,X_{1:j-1},Y_{1:j-1};X_j,Y_j)\\&=nI(E;X,Y),\\
      I(K;Y^n|X^n)=&\sum_{j=1}^nI(K;Y_j|X^n,Y_{1:j-1})\\&=\sum_{j=1}^nI(K,X_{1:j-1},X_{j+1:n},Y_{1:j-1};Y_j|X_j)\geq\\& \sum_{j=1}^nI(K,X_{1:j-1},Y_{1:j-1};Y_j|X_j)\\&=nI(E;Y|X)\end{align*}and similarly \begin{align*}
      I(K;X^n|Y^n)\geq nI(E;X|Y).\end{align*}
      Since $\lambda_2\leq 0,\lambda_3\leq 0$, we have $\lambda_2\frac{1}{n}I(K;Y^n|X^n)\leq I(E;Y|X)$ and $\lambda_3\frac{1}{n}I(K;X^n|Y^n)\leq I(E;X|Y)$. Therefore the maximum of $\lambda_1 u_1+\lambda_2 u_2 + \lambda_3 u_3+\lambda_4 u_4$ over $U(p)$ will be less than or equal to the maximum of the same expression over the first set of points.
  \item $\lambda_1\leq 0,\lambda_2\geq 0,\lambda_3\leq 0:$ If $\lambda_1+\lambda_2\geq 0$, we can write \begin{align*}&\lambda_1I(K;X^nY^n)+\lambda_2I(K;Y^n|X^n)+\lambda_3I(K;X^n|Y^n)\leq\\&
\lambda_1I(K;X^nY^n)+\lambda_2I(K;Y^n|X^n)=\\&
\lambda_1I(K;X^n)+(\lambda_1+\lambda_2)I(K;Y^n|X^n)\leq \\&
(\lambda_1+\lambda_2)I(K;Y^n|X^n)\leq (\lambda_1+\lambda_2)H(Y^n|X^n)
      \end{align*}
       Thus the maximum of $\lambda_1 u_1+\lambda_2 u_2 + \lambda_3 u_3+\lambda_4 u_4$ over $U(p)$ will be less than or equal to $(\lambda_1+\lambda_2)H(Y|X)$, which is equal to the maximum of the same expression over the second set of points.
       If $\lambda_1+\lambda_2\leq 0$, we can write
             \begin{align*}&\lambda_1I(K;X^nY^n)+\lambda_2I(K;Y^n|X^n)+\lambda_3I(K;X^n|Y^n)=\\&(\lambda_1+\lambda_2)I(K;X^nY^n)-\lambda_2I(K;X^n)+\lambda_3I(K;X^n|Y^n)\\&\leq 0.
              \end{align*}
       Thus the maximum of $\lambda_1 u_1+\lambda_2 u_2 + \lambda_3 u_3+\lambda_4 u_4$ over $U(p)$ will be zero.
   \item $\lambda_1\leq 0,\lambda_2\leq 0,\lambda_3\geq 0:$ This is similar to case 5.
   \item $\lambda_1\leq 0,\lambda_2\leq 0,\lambda_3\leq 0:$ This is similar to case 4.
   \item $\lambda_1\leq 0,\lambda_2\geq 0,\lambda_3\geq 0:$ If $\lambda_1+\lambda_2+\lambda_3\leq 0$
   \begin{align*}&
   \lambda_1I(K;X^nY^n)+\lambda_2I(K;Y^n|X^n)+\lambda_3I(K;X^n|Y^n)=\\&(\lambda_1+\lambda_2+\lambda_3)I(K;X^nY^n)-\lambda_2I(K;X^n)-\lambda_3I(K;Y^n)\\&\leq 0.
   \end{align*}
   Thus $K$ constant works here. If $\lambda_1+\lambda_2+\lambda_3\geq 0$ using Lemma \ref{lemma1E}
    \begin{eqnarray*}&
\lambda_1I(K;X^nY^n)+\lambda_2I(K;Y^n|X^n)\\&+\lambda_3I(K;X^n|Y^n)=\\&
(\lambda_1+\lambda_2+\lambda_3)I(K;X^nY^n)-\lambda_2I(K;X^n)\\&-\lambda_3I(K;Y^n)\leq\\&
(\lambda_1+\lambda_2+\lambda_3)I(K;X^nY^n)\\&-\lambda_2[I(K;X^nY^n)-H(Y^n|X^n)]_+\\&-\lambda_3[I(K;X^nY^n)-H(X^n|Y^n)]_+=\\&
n\bigg(\lambda_1\frac{I(K;X^nY^n)}{n}+\lambda_2\min(\frac{I(K;X^nY^n)}{n}, \\&H(Y|X))+\lambda_3\min(\frac{I(K;X^nY^n)}{n},H(X|Y))\bigg).
\end{eqnarray*}
Thus, the maximum of the original expression is less than or equal to
\begin{align*}
&\max_{0\leq t\leq H(X,Y)}\bigg(\lambda_1t+\lambda_2\min(t, H(Y|X))\\&+\lambda_3\min(t,H(X|Y))\bigg)=\\&\max_{0\leq t\leq \max(H(X|Y),H(Y|X))}\bigg(\lambda_1t+\lambda_2\min(t, H(Y|X))\\&+\lambda_3\min(t,H(X|Y))\bigg).
\end{align*}
Thus the maximum of $\lambda_1 u_1+\lambda_2 u_2 + \lambda_3 u_3+\lambda_4 u_4$ over $U(p)$ will be less than or equal to the maximum of the same expression over the fourth set of points.
\end{itemize}

\begin{lemma}\label{lemma1E} Given any three random variables $X,Y,K$ where $K$ is a function of $(X,Y)$, we have
$$I(K;X)\geq [H(K)-H(Y|X)]_+$$
$$I(K;Y)\geq [H(K)-H(X|Y)]_+$$
where $[x]_+$ is $0$ when $x$ is negative and $x$ when it is non-negative.
\end{lemma}
\emph{Proof:} We prove the first equation. The proof for the second one is similar. It suffices to show that
$I(K;X)\geq H(K)-H(Y|X)$, which is equivalent with $H(Y,X)\geq H(K,X)$ and obviously true.

\end{proof}
\begin{proof}[Proof of Theorem \ref{Thm2}] Take some $n$ and $p(k|x^n,y^n)$ and consider the 4-tuples $(u_1,u_2,u_3,u_4)$
\begin{align*}
&u_1=\frac{1}{n}H(K)\\&u_2=\frac{1}{n}H(K|X^n)\\&u_3=\frac{1}{n}H(K|Y^n)\\&u_4=\frac{1}{n}H(K|X^n,Y^n)
\end{align*}
Let $c=\frac{1}{n}H(K|X^n,Y^n)$. Let the random index $J$ be uniformly distributed on $\{1,2,3,...,n\}$ and independent of $(K,X^n,Y^n)$. Define the auxiliary random variables $E=(K,X_{1:J-1},Y_{1:J-1},J), X=X_J, Y=Y_J$. One can then verify that
      \begin{align*}&I(K;X^n,Y^n)=nI(E;X,Y),\\
      &I(K;Y^n|X^n)\geq nI(E;Y|X)\\
      &I(K;X^n|Y^n)\geq nI(E;X|Y).\end{align*}
Thus, $u_1=c+I(E;X,Y)$, $u_2\geq c+I(E;Y|X)$ and $u_3\geq c+I(E;X|Y)$ for some $p(e|x,y)$.
\end{proof}

\section*{Acknowledgment}
The authors would like to thank Prof. Raymond Yeung for many insightful discussions, his comments on the manuscript and continued support throughout the research. This research was partially supported by a grant from the University Grants Committee of the Hong Kong Special Administrative Region, China (Project No. AoE/E-02/08).
\appendix
\subsection{Using the cuts to write a converse}
In this appendix we use cuts between sources and sinks to write a converse for the network of Figure \ref{fig:f3N}. Since there are two sources and two sinks in this network there are more types of cuts to consider. Every cut divides the nodes of the network into two sets $\mathcal{A}$ and $\mathcal{A}^c$. We use the notation $cut$\emph{(sources in $\mathcal{A}$; sources in $\mathcal{A}^c$; sinks in $\mathcal{A}^c$)} to denote the edges of such a cut. For instance in Figure \ref{fig:f3N}, $\{4,2\}$ is $cut(s_1; s_2; t_1,t_2)$ meaning that edges $4$ and $2$ are the edges of a cut that has $s_1$ in $\mathcal{A}$,  $s_2$ in $\mathcal{A}^c$ and sinks $t_1,t_2$ in $\mathcal{A}^c$. Suppose we want to write the converse for an edge $e$ in $cut$\emph{(sources in $\mathcal{A}$; sources in $\mathcal{A}^c$; sinks in $\mathcal{A}^c$)}. If there is no source in $\mathcal{A}^c$, then we can write a converse as discussed earlier in equations (\ref{eqn:edgeconverse1}-\ref{eqn:edgeconverse3}). However if there is a source in $\mathcal{A}^c$, say $s_2$, we need to use a modified version of Lemma \ref{lemma1} used to bound the entropy of the random variable on an edge of the cut conditioned on a source that is in $\mathcal{A}$. The inequality of the lemma is weakened by adding the joint entropy of all the sources in $\mathcal{A}^c$ to one side of the inequality as shown below.

\emph{Lemma \ref{lemma1} [revisited]:} Take an arbitrary cut containing $e$ from the first source to the first sink, and let $Cut_x$ denote the sum of the capacities of the edges on this cut. Further assume that $s_2$ is in $\mathcal{A}^c$. Then $\frac{1}{n}H(K|X^n)$ must satisfy the following inequalities:
$$\frac{1}{n}H(K|X^n)\leq Cut_x-C_e+d_e+H(Y)-H(X)+k(\epsilon)$$
$$\frac{1}{n}H(K|X^n,Y^n)\leq Cut_x-C_e+d_e+H(Y)-H(X,Y)+k(\epsilon)$$
for some functions $k(\epsilon)$ that converges to zero as $\epsilon$ converges to zero.

\begin{proof} Let $Q$ denote the collection of random variables passing over the edges of the cut (except $e$). Clearly $\frac{1}{n}H(Q)\leq Cut_x-C_e+m\epsilon$ where $m$ is the number of edges in the graph. Since $(Q,K)$ is the collection of the random variables passing the edges of the cut, $X^n$ should be recoverable from $(Q,K,Y^n)$ with probability of error less than or equal to $\epsilon$. Thus, by Fano's inequality $\frac{1}{n}H(X^n|Q,K,Y^n)\leq k_1(\epsilon)$ for some function $k_1(\epsilon)$ that converges to zero as $\epsilon$ converges to zero. We have
\begin{align*}\frac{1}{n}H(K|X^n)&\leq \frac{1}{n}H(Q,K,Y^n|X^n)\\&=\frac{1}{n}H(Q,K,Y^n,X^n)-\frac{1}{n}H(X^n)\\&\leq\frac{1}{n}H(Q)+\frac{1}{n}H(K)+H(Y)\\&~~~+\frac{1}{n}H(X^n|Q,K,Y^n)-H(X)\\&\leq Cut_x-C_e+H(Y)\\&~~~+m\epsilon+d_e-H(X)+k_1(\epsilon).\end{align*}
We get the first inequality by setting $k(\epsilon)=k_1(\epsilon)+m\epsilon$.
For the second inequality note that
\begin{align*}\frac{1}{n}H(K|X^n,Y^n)&\leq \frac{1}{n}H(Q,K|X^n,Y^n)\\&=\frac{1}{n}H(Q,K,Y^n,X^n)-\frac{1}{n}H(X^n,Y^n)\\&\leq\frac{1}{n}H(Q)+\frac{1}{n}H(K)+H(Y)\\&~~~+\frac{1}{n}H(X^n|Q,K,Y^n)-H(X,Y)\\&\leq Cut_x-C_e+H(Y)\\&~~~+m\epsilon+d_e-H(X,Y)+k_1(\epsilon).\end{align*}
\end{proof}

We can now write down the converse using the edge-cuts. We proceed in a similar fashion that we did in deriving equations (\ref{eqn:edgeconverse1}-\ref{eqn:edgeconverse3}) using Lemma \ref{lemma1} (revisited) and Theorem \ref{Thm2}. Lemma \ref{lemma1} (revisited) gives us upper bounds on the elements of the uncertainty vector, whereas Theorem \ref{Thm2} gives us lower bounds on these elements.

Cuts that have edge 2:
\begin{align*}
&d_2=I(E_2;XY)\leq C_2\\
&C_2+C_4-C_2+d_2+H(Y)-H(X)\geq I(E_2;Y|X)\\
&C_2+C_4-C_2+d_2+H(Y)-H(X,Y)\geq 0\\&~~~~\mbox{because }\{2,4\}\mbox{ is } cut(s_1; s_2; t_1, t_2)\\
&C_2+C_3+C_4+C_5-C_2+d_2-H(X,Y)\geq 0\\&~~~~\mbox{because }\{2,3,4,5\}\mbox{ is } cut(s_1, s_2;\emptyset; t_1, t_2)\\
\end{align*}
for some $p(e_2|x,y)$.

Cuts that have edge 3:
\begin{align*}
&d_3=I(E_3;XY)\leq C_3\\
&C_3+C_5-C_3+d_3+H(X)-H(Y)\geq I(E_2;X|Y)\\
&C_3+C_5-C_3+d_3+H(X)-H(X,Y)\geq 0\\&~~~~\mbox{because }\{3,5\}\mbox{ is } cut(s_2; s_1; t_1, t_2)\\
&C_2+C_3+C_4+C_5-C_3+d_3-H(X,Y)\geq 0\\&~~~~\mbox{because }\{2,3,4,5\}\mbox{ is } cut(s_1, s_2;\emptyset; t_1, t_2)\\
\end{align*}
for some $p(e_3|x,y)$.

Cuts that have edge 4:
\begin{align*}
&d_4=I(E_4;XY)\leq C_4\\
&C_2+C_4-C_4+d_4+H(Y)-H(X)\geq I(E_4;Y|X)\\
&C_2+C_4-C_4+d_4+H(Y)-H(X,Y)\geq 0\\&~~~~\mbox{because }\{2,4\}\mbox{ is } cut(s_1; s_2; t_1, t_2)\\
&C_2+C_3+C_4+C_5-C_4+d_4-H(X,Y)\geq 0\\&~~~~\mbox{because }\{2,3,4,5\}\mbox{ is } cut(s_1, s_2;\emptyset; t_1, t_2)\\
&C_4+C_5+C_6-C_4+d_4-H(X,Y)\geq 0\\&~~~~\mbox{because }\{4,5,6\}\mbox{ is } cut(s_1, s_2;\emptyset; t_1, t_2)\\
&C_4+C_7-C_4+d_4-H(X)\geq I(E_4;Y|X)\\&~~~~\mbox{because }\{4,7\}\mbox{ is } cut(s_1, s_2;\emptyset; t_1)\\
\end{align*}
for some $p(e_4|x,y)$.

Cuts that have edge 5:
\begin{align*}
&d_5=I(E_5;XY)\leq C_5\\
&C_3+C_5-C_5+d_5+H(X)-H(Y)\geq I(E_5;X|Y)\\
&C_3+C_5-C_5+d_5+H(X)-H(X,Y)\geq 0\\&~~~~\mbox{because }\{3,5\}\mbox{ is } cut(s_2; s_1; t_1, t_2)\\
&C_2+C_3+C_4+C_5-C_5+d_5-H(X,Y)\geq 0\\&~~~~\mbox{because }\{2,3,4,5\}\mbox{ is } cut(s_1, s_2;\emptyset; t_1, t_2)\\
&C_4+C_5+C_6-C_5+d_5-H(X,Y)\geq 0\\&~~~~\mbox{because }\{4,5,6\}\mbox{ is } cut(s_1, s_2;\emptyset; t_1, t_2)\\
&C_5+C_8-C_5+d_5-H(Y)\geq I(E_5;X|Y)\\&~~~~\mbox{because }\{5,8\}\mbox{ is } cut(s_1, s_2;\emptyset; t_2)\\
\end{align*}
for some $p(e_5|x,y)$.

Since the capacities of edges 6, 7 and 8 are all the same, we can assume that they are all carrying the same message. Therefore we can compute the uncertainty of the message on edge 6 by looking at cuts that include edge 7 or 8.
\begin{align*}
&d_6=I(E_6;XY)\leq C_6\\
&C_4+C_5+C_6-C_6+d_6-H(X,Y)\geq 0\\&~~~~\mbox{because }\{4,5,6\}\mbox{ is } cut(s_1, s_2;\emptyset; t_1, t_2)\\
&C_4+C_6-C_6+d_6+H(Y)-H(X)\geq I(E_6;Y|X)\\
&C_4+C_6-C_6+d_6+H(Y)-H(X,Y)\geq 0\\&~~~~\mbox{because }\{4,6\}\mbox{ is } cut(s_1; s_2; t_1, t_2)\\
&C_5+C_6-C_6+d_6+H(X)-H(Y)\geq I(E_6;X|Y)\\
&C_5+C_6-C_6+d_6+H(X)-H(X,Y)\geq 0\\&~~~~\mbox{because }\{5,6\}\mbox{ is } cut(s_2; s_1; t_1, t_2)\\
&C_4+C_6-C_6+d_6-H(X)\geq I(E_6;Y|X)\\&~~~~\mbox{because }\{4,7\}\mbox{ is } cut(s_1, s_2;\emptyset; t_1)\\
&C_5+C_6-C_6+d_6-H(Y)\geq I(E_6;X|Y)\\&~~~~\mbox{because }\{5,8\}\mbox{ is } cut(s_1, s_2;\emptyset; t_2)\\
&C_4+C_5+C_7+C_8-C_6+d_6-H(X,Y)\geq 0\\&~~~~\mbox{because }\{4,5,7,8\}\mbox{ is } cut(s_1, s_2;\emptyset; t_1, t_2)\\
\end{align*}
for some $p(e_6|x,y)$.
After simplification and removal of redundant equations and noting that $C_6=C_7=C_8$, these inequalities can be written as follows:
\begin{align}
&I(E_2;X,Y)\leq C_2\label{eqn:converseeq-first}\\
&C_4\geq H(X,Y|E_2)-H(Y)\\
&C_3+C_4+C_5\geq H(X,Y|E_2)\\
&~~~~~~~~\mbox{\emph{From equations for edge 2}} \nonumber\\
&I(E_3;X,Y)\leq C_3\\
&C_5\geq H(X,Y|E_3)-H(X)\\
&C_2+C_4+C_5\geq H(X,Y|E_3)\\
&~~~~~~~~\mbox{\emph{From equations for edge 3}} \nonumber\\
&I(E_4;X,Y)\leq C_4\\
&C_2\geq H(X,Y|E_4)-H(Y)\\
&C_2+C_3+C_5\geq H(X,Y|E_4)\\
&C_5+C_6\geq H(X,Y|E_4)\\
&C_6\geq H(X|E_4)\\
&~~~~~~~~\mbox{\emph{From equations for edge 4}} \nonumber\\
&I(E_5;X,Y)\leq C_5\\
&C_3\geq H(X,Y|E_5)-H(X)\\
&C_2+C_3+C_4\geq H(X,Y|E_5)\\
&C_4+C_6\geq H(X,Y|E_5)\\
&C_6\geq H(X|E_5)\\
&~~~~~~~~\mbox{\emph{From equations for edge 5}} \nonumber
&I(E_6;X,Y)\leq C_6\\
&C_4+C_5\geq H(X,Y|E_6)\\
&C_4\geq H(X|E_6)\\
&C_5\geq H(Y|E_6)\label{eqn:converseeq-last}
\\
&~~~~~~~~\mbox{\emph{From equations for edge 6}} \nonumber
\end{align}
for some $p(e_2,e_3,e_4,e_5,e_6|x,y)$.

We claim that the minimum possible value of $C_6$ in this converse is less than or equal to $I(X;Y)$ if we restrict ourselves to networks where $C_2+C_4=H(X|Y)$. This is because the choice of $C_2=0$, $C_3=H(Y)$, $C_4=H(X|Y)$, $C_5=H(X,Y)$ and $C_6=I(X;Y)$ is a valid point in this converse region. To see this take $E_6$ in a way that $E_6\rightarrow X\rightarrow Y$ forms a Markov chain, and furthermore $p(e_6|x)\sim p(y|x)$. Take $E_4$ in a way that $E_4\rightarrow X\rightarrow Y$ forms a Markov chain, and furthermore $I(E_4;X)=H(X|Y)$. Take $E_5=(X,Y)$, $E_3=Y$ and $E_2=constant$. To verify these equations, it is useful to note that since $C_5=H(X,Y)$ those equations involving $C_5$ will be automatically satisfied. Because $E_6\rightarrow X\rightarrow Y$ forms a Markov chain and $p(e_6|x)\sim p(y|x)$, we have $I(E_6;X,Y)=I(E_6;X)=I(Y;X)$.

\end{document}